\documentclass{article}
\usepackage{tikz}
\usepackage{fullpage}
\usepackage{amsmath}
\usepackage{booktabs,amsthm}
\usepackage{amssymb}
\usepackage{amsfonts}
\usepackage{hyperref}
\usepackage{color}
\usepackage{xcolor}
\usepackage{mathbbol}
\usepackage{algpseudocode}
\usepackage{multicol}
\usepackage{relsize}
\usepackage{algorithm2e}
\usepackage{marginnote}
\usepackage{hyperref}
\usepackage{cite}
\usepackage{caption}
\usetikzlibrary{
            matrix,
            positioning,
        }
\newcommand{\sac}[1]{%
  \begingroup
  \begingroup\lccode`~=`, \lowercase{\endgroup
    \edef~{\mathchar\the\mathcode`, \penalty0 \noexpand\hspace{0pt plus 1em}}%
  }\mathcode`,="8000 #1%
  \endgroup
}

\newcommand{\ignore}[1]{}

\usepackage{varwidth}
\DeclareCaptionFormat{myformat}{%
  \begin{varwidth}{\linewidth}%
    \centering
    #1#2#3%
  \end{varwidth}%
}

\newtheorem{theorem}{Theorem}

\newtheorem{lemma}[theorem]{Lemma}

\newtheorem{definition}[theorem]{Definition}

\renewcommand{\Pr}{{\bf Pr}}
\newcommand{\E}{{\bf E}}

\title{Optimal Deterministic Group Testing Algorithms \\ to Estimate the Number of Defectives}
\author{{\bf  Nader H. Bshouty}  \\ Dept. of Computer Science\\ Technion,  Haifa, 32000\\ \ \\
{\bf  Catherine A. Haddad-Zaknoon}  \\ Dept. of Computer Science\\ Technion,  Haifa, 32000\\
}
\begin{document}

\maketitle

\abstract{We study the problem of estimating the number of defective items $d$ within a pile of $n$ elements up to a multiplicative factor of $\Delta>1$, using deterministic group testing algorithms. We bring lower and upper bounds on the number of tests required in both the adaptive and the non-adaptive deterministic settings given an upper bound $D$ on the defectives number. For the adaptive deterministic settings, our results show that, any algorithm for estimating the defectives number up to a multiplicative factor of $\Delta$ must make at least $\Omega \left((D/\Delta^2)\log (n/D) \right )$ tests. This extends the same lower bound achieved in~\cite{ALA17} for non-adaptive algorithms. Moreover, we give a polynomial time adaptive algorithm that shows that our bound is tight up to a small additive term.

For non-adaptive algorithms, an upper bound of $O((D/\Delta^2)$ $(\log (n/D)+\log \Delta) )$ is achieved by means of non-constructive proof. This improves the lower bound $O((\log D)/(\log\Delta))D\log n)$ from~\cite{ALA17} and matches the lower bound up to a small additive term.

In addition, we study polynomial time constructive algorithms. We use existing polynomial time constructible \emph{expander regular bipartite graphs}, \emph{extractors} and \emph{condensers} to construct two polynomial time algorithms. The first algorithm makes $O((D^{1+o(1)}/\Delta^2)\cdot \log n)$ tests, and the second makes $(D/\Delta^2)\cdot quazipoly$ $(\log n)$ tests. This is the first explicit construction with an almost optimal test complexity.}

\section{Introduction}
The problem of \emph{group testing} is the problem of identifying or, in some cases, examining the properties of a small amount of items known as  \emph{defective} items within a pile of elements using \emph{group tests}.  Let $X$ be a set of $n$ items, and let $I\subseteq X$ be the set of defective items. A \emph{group test} is a subset $Q\subseteq X$ of items. The result of the test $Q$ with respect to $I$ is defined by $Q(I):=1$ if $Q\cap I \neq \emptyset$ and $Q(I):=0$ otherwise. While the defective set $I$ is unknown to the algorithm, in many cases we might be interested in finding the size of the defective set $|I|$,  or at least an estimation of that value with a minimum number of tests. \ignore{The algorithms considered in this paper might operate in \emph{stages} or \emph{rounds}. In each round, the tests are defined in advance and tested in a single parallel step. Tests on some round might depend on the test results of the preceding rounds. A single round algorithm is called \emph{non-adaptive} algorithm, while a multi-round algorithm is called \emph{adaptive algorithm}}

Group testing was originally proposed as a potential solution for economising mass blood testing during WWII~\cite{Dor43}. Since then, group testing approach has been diversely applied in a wide area of practical applications including DNA library screening~\cite{DH2_00}, product testing quality control\cite{SG59}, file searching in storage systems~\cite{KS64}, sequential screening of experimental variables~\cite{Li62}, efficient contention algorithms for MAC~\cite{KS64,JW85}, data compression~\cite{HL02}, and computations in data stream model~\cite{CM05}. Recently, during the COVID-19 pandemic outbreak, a number of researches adopted the group testing paradigm not only to accelerate  mass testing process, but also to dramatically reduce the number of kits required for testing due to severe shortages in the testing kits supply~\cite{YK20,GC20, MCR20}.

While an up-front knowledge of the value of $d$ or at least an upper bound on it is required in many of the algorithms aimed at identifying the defective items, estimating or finding the number of defectives is an interesting problem on its own as well.  Defectives estimation via group testing has been applied vastly in biological and medical applications~\cite{CWH90, Sw85,Tk62,WHB80, GH89}. In~\cite{Tk62}, for example, group testing algorithms are  used to estimate aster-yellow virus transmitters proportion over the organisms in a natural population of leafhoppers.  Similarly, in~\cite{WHB80}, the authors estimate the infection rate of the yellow-fever virus in mosquito population using group testing methods. On the other hand,  in~\cite{GH89}, group-testing-based estimation of rare diseases prevalence is employed not only for its effectiveness but also because it naturally preserves individual anonymity of the subjects.

Algorithms dedicated for this task  might operate in \emph{stages} or \emph{rounds}. In each round, the tests are defined in advance and tested in a single parallel step. Tests on some round might depend on the test results of the preceding rounds. A single round algorithm is called \emph{non-adaptive} algorithm, while a multi-round algorithm is called \emph{adaptive algorithm}.

In recent years, there has been an increasing interest in the problem of estimating the number of defective items via group testing\cite{BBHKS18, Bshouty19, CWH90, CX14, DS10, DS09,FJOPS16, RT18}. The target in some of these papers is to find an estimation $\hat{d}$ within an additive factor of $\epsilon< 1$ such that $(1-\epsilon)d \leq \hat{d}\leq (1+\epsilon)d$. For randomized adaptive algorithms we have the following results. Falhatgar et.al.~\cite{FJOPS16} give a randomised adaptive algorithm that estimates $d$ using  $2\log\log d +O(1/\epsilon^2\log 1/\delta)$  queries in expectation where $\delta$ is the failure probability of the algorithm. Bshouty et. al.~\cite{BBHKS18} modified this result and gave an algorithm that uses $(1-\delta) \log\log d +O((1/\epsilon^2)\log 1/\delta)$ expected number of queries. Moreover, they proved a lower bound of $(1-\delta)\log\log d +\Omega((1/\epsilon)\log(1/\delta))$ queries.

For randomized non-adaptive algorithms with constant estimation, Damaschke and Sheikh Muhammad give in~\cite{DS09} a randomized non-adaptive algorithm that makes $O((\log(1/\delta))\log n)$ tests and in \cite{Bshouty19}, Bshouty gives the lower bound $\Omega(\log n/\log\log n)$.

In this paper, we are interested in {\emph deterministic} adaptive and non-adaptive algorithms that estimate the defective items set size $d$ up to a multiplicative factor of $\Delta >1$. Formally, let $|I|:=d$ and let $D \geq d$. We say that a deterministic algorithm $\cal A$ estimates $d$ up to a multiplicative factor of $\Delta$ if, given $D$ as an input to the algorithm, it evaluates an estimation $\hat{d}$ such that $d/\Delta \leq \hat{d}\leq d\Delta$. Bshouty et al. show in~\cite{BBHKS18} that, if no upper bound $D$ is given to the algorithm, then any deterministic adaptive algorithm (and therefore also non-deterministic algorithm) for this problem must make at least $\Omega(n)$ tests. This is equivalent to testing all the items. This justifies the fact that any non-trivial efficient algorithm must have some upper bound $D$ for $d$.

Agarwal et.al. ~\cite{ALA17} consider this problem. They first give the lower bound of~$\Omega ((D/\Delta^2)$ $\log ({n}/{D}))$ queries for any non-adaptive deterministic algorithm. Moreover, using a non-constructive proof, they give an upper bound of $O\left ((({\log D})/({\log \Delta}))D\log n \right )$ queries.

We further investigate this problem. We bring new lower and upper bounds on the number of tests required both in adaptive and non-adaptive deterministic algorithms. For the adaptive deterministic settings, our results show that, any algorithm for estimating the defectives number up to a multiplicative factor of $\Delta$ must make at least $\Omega \left((D/\Delta^2)\log (n/D) \right )$ tests. This extends the same lower bound achieved in~\cite{ALA17} for non-adaptive algorithms. Furthermore, we give a polynomial time adaptive algorithm that shows that our bound is tight up to a small additive term.

For non-adaptive algorithms, we achieve an upper bound of $O((D/\Delta^2)$ $(\log (n/D)+\log \Delta))$  by means of non-constructive proof.  This improves the lower bound $O((\log D)/(\log\Delta))D\log n)$ from~\cite{ALA17}, and matches the lower bound up to a small additive term.

We then study polynomial time constructive algorithms. For this task, we use existing polynomial time constructible \emph{expander regular bipartite graphs}, \emph{extractors} and \emph{condensers} to construct two polynomial time algorithms. The first algorithm makes $O((D^{1+o(1)}/\Delta^2)\cdot \log n)$ tests, and the second makes $(D/\Delta^2)\cdot quazipoly$ $(\log n)$ tests. To the best of our knowledge, this is the first explicit construction with an almost optimal test complexity. Our results are summarised in Table~\ref{resTbl}.

\captionsetup{justification=centering, format=myformat}
{\renewcommand{\arraystretch}{1.6}
\begin{table}[h]
\begin{center}
\begin{tabular}{ |c| c|c|c|c| }
\hline
Bounds & Adaptive/ & Result &Explicit/ &Ref.\\
& Non-Adapt.  &  & Non-Expl. &\\
 \hline
\hline
 Lower B. & Non-Adapt. & $\frac{D}{\Delta^2}\log \frac{n}{D}$  &-& \cite{ALA17}\\
 \hline
 Lower B. & Adaptive & $\frac{D}{\Delta^2}\log \frac{n}{D}$ &- & Ours\\
 \hline
 \hline
 Upper B. &Adaptive & $\frac{D}{\Delta^2}\left(\log \frac{n}{D}+\log \Delta\right) $ &Explicit & Ours \\
 \hline
 Upper B. & Non-Adapt. & $\frac{\log D}{\log \Delta} D\log n$ &Non-Expl. &\cite{ALA17}\\
 \hline
 Upper B. & Non-Adapt. & $\frac{D}{\Delta^2}\left(\log \frac{n}{D}+\log \Delta \right)$ &Non-Expl. &Ours\\
  \hline
 Upper B. & Non-Adapt. & $\frac{D^{1+o(1)}}{\Delta^2}\log n $ &Explicit\footnotemark &Ours\\
  \hline
 Upper B. & Non-Adapt. & $\frac{D}{\Delta^2}\cdot{\rm Quazipoly}(\log n)$ &Explicit &Ours\\
 \hline
 \end{tabular}
\end{center}
\caption{Upper and lower bounds on the number of tests required\newline for estimating defectives in deterministic group testing.}
\label{resTbl}
\end{table}
}

 \ignore{
\color{blue}  We will consider the two cases: (1) $\Delta=1+\Theta(1)$ and (2) $\Delta=\omega(1)$. For both cases,  achieve the upper bound
$$O\left(\frac{\log D}{\log \Delta}D\log n\right)$$
and the lower bound
$$\Omega\left(\frac{D}{\Delta^2}\log \frac{n}{D}\right).$$
When $\Delta=1+\Theta(1)$ (case (1)) the bounds are
$$O\left({(\log D)}D\log n\right)$$
and $$\Omega\left({D}\log \frac{n}{D}\right).$$
We will show that the lower bounds in both cases are tight and prove the same lower bounds for any deterministic adaptive algorithm. We also give an adaptive algorithm that achieve those bounds.
\color{black}
}

\section{Definitions and Preliminary Results}
In this section, we give some notations and definition that will be used in this paper.

Let $X=[n]:=\{1,\cdots, n\}$ be a set of \emph{items}. Let $I\subseteq X$ be a set of \emph{defective} items, and let $d$ denote its size, i.e. $d=|I|$. In the group testing settings, a \emph{test} is a subset $Q\subseteq X$ of items. An answer to a test $Q$ with respect to the defective items set $I$, is denoted by $Q(I)$, such that $Q(I):=1$ if $Q\cap I \neq \emptyset$ and $0$ otherwise. We denote by ${\cal O}_I$ an \emph{oracle} that for a test $Q$ returns $Q(I)$.

Let ${\cal A}$ be an algorithm with an access to ${\cal O}_I$, and let $d=|I|$. We say that the algorithm ${\cal A}$ {\it estimates $d$ up to a multiplicative factor of} $\Delta$, if ${\cal A}$ gets as an input an upper bound $D\geq d$ and a parameter $\Delta>1$, and outputs $\hat d$ such that $d/\Delta \leq \hat d \leq d\Delta$.
We say that ${\cal A}$ is an \emph{adaptive} algorithm, if its queries depend on the result of previous queries, and \emph{non-adaptive} if its queries are independent of previous ones and therefore, can be executed in a single parallel step.
\footnotetext[1]{This result is true for $\Delta>C$ for some constant $C$. See section~\ref{lsection}.}
We may assume that $D\geq \Delta^2$, otherwise, the algorithm trivially outputs $\hat d = D/\Delta$. We note here that $\Delta\ge 1+\Omega(1)$, that is, it is greater than a constant that is greater than $1$ and it may depend\footnote{For example $\Delta=\log\log n+\log D$} on $n$ and/or $D$. This is implicit in~\cite{ALA17} and is also constrained in this paper. It is also interesting to investigate this problem when $\Delta=1+o(n)$ where $o()$ (small $o$) is with respect to $D$ and/or $n$.

We will use the following 
\begin{lemma}\label{Chernoff}{\bf Chernoff's Bound}. Let $X_1,\ldots, X_m$ be independent random variables taking values in $\{0, 1\}$. Let $X=\sum_{i=1}^mX_i$ denotes their sum and $\mu = \E[X]$ denotes the sum's expected value. Then
\begin{eqnarray}\label{Chernoff1}\Pr[X>(1+\lambda)\mu]\le \left(\frac{e^{\lambda}}{(1+\lambda)^{(1+\lambda)}}\right)^{\mu}\le e^{-\frac{\lambda^2\mu}{2+\lambda}}\le  \begin{cases} e^{-\frac{\lambda^2\mu}{3}} &\mbox{if\ } 0< \lambda\le 1 \\
e^{-\frac{\lambda \mu}{3}} & \mbox{if\ } \lambda>1 \end{cases} .\label{Chernoff1}
\end{eqnarray}
In particular,
\begin{eqnarray}\label{Chernoff2}\Pr[X>\Lambda]\le \left(\frac{e\mu}{\Lambda}\right)^{\Lambda}.\end{eqnarray}

For $0\le \lambda\le 1$ we have
\begin{eqnarray}\label{Chernoff3}
\Pr[X<(1-\lambda)\mu]\le \left(\frac{e^{-\lambda}}{(1-\lambda)^{(1-\lambda)}}\right)^{\mu}\le e^{-\frac{\lambda^2\mu}{2}}.\label{Chernoff2}
\end{eqnarray}

\end{lemma}
Moreover, we will often use the inequality

\begin{eqnarray}
\label{eqbinom}\label{eqbinomSum}
\left(\frac{n}{k}\right)^k\le
{n\choose k}\le \sum_{i=0}^{k}{n\choose i}\le \left(\frac{en}{k}\right)^k,
\end{eqnarray}

\section{Upper Bound for Non-Adaptive Deterministic Algorithms}\label{Msection}
In this section, we give the upper bound for deterministic non-adaptive algorithm that estimates $d$ up to a multiplicative factor of $\Delta$. We prove:

\begin{theorem}
\label{thmUBE}
Let $D$ be some upper bound on the number of defective items $d$ and $\Delta >1$. Then, there is a deterministic non-adaptive algorithm that makes
$$O\left (\frac{D}{\Delta^2}\left(\log \frac{n}{D} +\log\Delta\right)\right )$$
tests and outputs $\hat{d}$ such that $\frac{d}{\Delta} \leq \hat{d} \leq d\Delta.$
\end{theorem}

To prove the Theorem we need the following:
\begin{lemma}\label{fest}  Let $\Delta>1$ and $\ell\ge 2\Delta^2$. There is a non-adaptive deterministic algorithm that makes
$$t=O\left(\frac{\ell}{\Delta^2}\left(\log\frac{n}{\ell}+\log \Delta\right)\right)$$ tests such that, 
\begin{enumerate}
\item If the number of defectives $d$ is less than $\ell/\Delta^2$, it outputs $0$. 
\item If it is greater than $\ell/\Delta$, it outputs $1$.
\end{enumerate}
\end{lemma}
\begin{proof} We choose a constant $c$ such that $(1-\Delta^2/(c\ell))^{\ell/\Delta^2}=1/e$. Note that
$$\left(1-\frac{\Delta^2}{2\ell}\right)^{\ell/\Delta^2}\ge 1-\frac{\Delta^2\ell}{2\ell\Delta^2}= \frac{1}{2}> \frac{1}{e}$$ and
$$\left(1-\frac{2\Delta^2}{\ell}\right)^{\ell/\Delta^2}= \left(\left(1-\frac{2\Delta^2}{\ell}\right)^{\frac{\ell}{2\Delta^2}}\right)^{2}\le \frac{1}{e^2}<\frac{1}{e}.$$ Therefore, such $c$ exists and we have $1/2\le c\le 2$.

Consider a test $Q\subseteq [n]$ chosen at random where each item $i\in[n]$ is chosen to be in $Q$ with probability $\Delta^2/(c\ell)$. Let $I$ be the set of defective items such that $|I|=d$, and let $Q(I)$ be the result of the test $Q$ with respect to the set $I$.  Then,
\begin{equation}
\label{eq1}
\Pr[Q(I)=0]=\left(1-\frac{\Delta^2}{c\ell}\right)^d.
\end{equation}
If $d\le \ell/\Delta^2$,
\begin{equation}
\label{eq2}
\Pr[Q(I)=0]\ge \left(1-\frac{\Delta^2}{c\ell}\right)^{\ell/\Delta^2}=e^{-1},
\end{equation}
if $d= 2\ell/\Delta^2$,
\begin{equation}
\label{eq25}
\Pr[Q(I)=0]= \left(1-\frac{\Delta^2}{c\ell}\right)^{2\ell/\Delta^2}=e^{-2},
\end{equation}
and if $d= \ell/\Delta$, we get:
\begin{equation}
\label{eq3}
\Pr[Q(I)=0]= \left(\left(1-\frac{\Delta^2}{c\ell}\right)^{\frac{\ell}{\Delta^2}}\right)^{\Delta}= e^{-\Delta}.
\end{equation}

Let $Q_1, Q_2,\ldots, Q_t$ be a sequence of $t$ i.i.d tests such that
$$t=\frac{c'\ell}{(\Delta-1)^2}\ln\frac{c''\Delta^2n}{\ell}$$
where $c'=54e^2$ and $c''=4e$.

\ignore{\color{blue}
\begin{eqnarray*}
t &=&\frac{8e}{(1-\Delta^{-1})^2\log e} + \frac{8e\ell}{(\Delta-1)^2 \log e}\log \left( \frac{e\Delta^2n}{\ell}\right ) \\
& =&\frac{8e}{(1-\Delta^{-1})^2\log e}+\frac{16e\ell  }{\log e} \cdot \frac{\log \Delta}{(\Delta -1)^2}+ \frac{8e}{ \log e}\cdot \frac{\ell}{(\Delta-1)^2} \log \left ( \frac{en}{\ell}\right )\\
&=&O\left(  \frac{\ell}{\Delta^2}\left (\log \left(\frac{n}{\ell} \right) + \log \Delta \right )\right)
\end{eqnarray*}
Nader: what is your favourable form of $t$?
\begin{equation}
\label{t_Final}
t=? =O\left(\frac{\ell}{(\Delta-1)^2}\log\frac{\Delta^2n}{\ell}\right),
\end{equation}
\color{black}}
Let $$\eta=e^{-1}\left(\frac{1}{2}+\frac{1}{2\Delta}\right).$$
Consider the following two events:
\begin{enumerate}
\item $A$: There is a set of defectives $I$ of size $|I|\le \ell/\Delta^2$ such that the number of tests with $0$ answer is less than $\eta t$.
\item $B$: There is a set of defectives $J$ of size $|J|> \ell/\Delta$ such that the number of tests with $0$ answer is at least $\eta t$.
\end{enumerate}
Notice that, to prove the lemma it is enough to prove that $\Pr[A\vee B]<1$. We will show that $\Pr[A],\Pr[B]< 1/2$.

Let $X_1, \ldots, X_t$ be random variables such that $X_i=1$ if and only if $Q_i(I)=0$. Let $X$ be the number of tests that yield the result $0$. Therefore, $X=\sum_{i=1}^t X_i$ and define $\mu:= \E[X]$.

If $|I|=d\leq \ell/\Delta^2$, then $\mu=t\cdot \E[X_i]=t\cdot\Pr[X_i =1]$. By (\ref{eq2}) we have
\begin{equation}
\mu = \E[X] \geq t\cdot e^{-1}.
\end{equation}
By (\ref{Chernoff3}) in Lemma~\ref{Chernoff}, for $\lambda = 1/2-1/(2\Delta)$ we have
$$\Pr[X\le \eta t]=\Pr[X\leq (1-\lambda)te^{-1}]\leq \Pr[X\leq (1-\lambda)\mu] \leq e^{-\frac{\lambda^2\mu}{2}}\leq e^{-\frac{(1-\Delta^{-1})^2t}{8e}}. $$
Using this result, equations (\ref{eqbinomSum}) and the union bound, we can conclude that
\begin{eqnarray*}
\label{AProb}\Pr[A]&\le& \left(\sum_{i=0}^{\ell/\Delta^2} {n\choose i}\right) e^{-\frac{(1-\Delta^{-1})^2t}{8e}}\le \left(\frac{e\Delta^2n}{\ell}\right)^{\frac{\ell}{\Delta^2}}e^{-\frac{(1-\Delta^{-1})^2t}{8e}}\\ &=& \left(\frac{e\Delta^2n}{\ell}\right)^{\frac{\ell}{\Delta^2}}e^{-\frac{c'\ell}{8e\Delta^2}\ln \frac{c''\Delta^2n}{\ell}}=
\left(\frac{e\Delta^2n}{\ell}\right)^{\frac{\ell}{\Delta^2}}\left(\frac{c''\Delta^2n}{\ell}\right)^{-\frac{c'\ell}{\Delta^2}}< \frac{1}{2}.
\end{eqnarray*}

\color{black}
On the other hand, for the event $B$, we have two cases.

\noindent 
{\bf Case I.} $1<\Delta\le 2$.

If there is a set of defectives $J$ of size $|J|> \ell/\Delta$ such that more than $\eta t$ of the tests yield the answer~$0$, then there is a set of defectives $J'$ of size $|J'|= \ell/\Delta$ such that more than $\eta t$ of the tests answers are~$0$. Denote by $B'$ the latter event. Then, by (\ref{eq3}) we have $\mu=\E[X]=e^{-\Delta}t$ and for $\lambda= (e^{\Delta-1}-1)/2\ge (\Delta-1)/2$, $\eta'=(e^{-1}+e^{-\Delta})/2\le \eta$ we get
\begin{eqnarray*}
\Pr[B]\le \Pr[B']&\le& {n\choose \ell/\Delta}\Pr\left[X\ge \eta t\right]\le {n\choose \ell/\Delta}\Pr\left[X\ge \eta' t\right]\\
&=&{n\choose \ell/\Delta}\Pr\left[X\ge \left(1+\lambda\right)\mu\right]\\
&\le& \left(\frac{e\Delta n}{\ell}\right)^{\frac{\ell}{\Delta}} \Pr\left[X\ge \left(1+\lambda\right)\mu \right]
\end{eqnarray*}
If $1< \Delta\le 2$ then $0\le \lambda\le 1$ and then by (\ref{Chernoff1}) in Lemma~\ref{Chernoff}, we have
\begin{eqnarray*}
\left(\frac{e\Delta n}{\ell}\right)^{\frac{\ell}{\Delta}} \Pr\left[X\ge \left(1+\lambda\right)\mu \right]&\le& \left(\frac{e\Delta n}{\ell}\right)^{\frac{\ell}{\Delta}} e^{-\lambda^2\mu/3}\\
&\le& \left(\frac{e\Delta n}{\ell}\right)^{\frac{\ell}{\Delta}} e^{-(\Delta-1)^2\mu/12}\ \ \ \mbox{since}\ \lambda\ge (\Delta-1)/2
\\
&=& \left(\frac{e\Delta n}{\ell}\right)^{\frac{\ell}{\Delta}} e^{-(\Delta-1)^2e^{-\Delta}t/12}\\
&=& \left(\frac{e\Delta n}{\ell}\right)^{\frac{\ell}{\Delta}} \left(\frac{c''\Delta^2n}{\ell}\right)^{-c'\ell e^{-\Delta}/12}\\
&\le& \left(\frac{2e n}{\ell}\right)^{\ell} \left(\frac{c''n}{\ell}\right)^{-(c' e^{-2}/12)\ell}<\frac{1}{2}\ \ \ \ \ \ \ \ 1\le \Delta<2\\
\end{eqnarray*}

\noindent
{\bf Case II.} $\Delta> 2$.

In this case we have $\ell/\Delta>2\ell/\Delta^2$. Therefore, if there is a set of defectives $J$ of size $|J|> \ell/\Delta$ such that more than $\eta t$ of the tests yield the answer~$0$, then there is a set of defectives $J'$ of size $|J'|= 2\ell/\Delta^2$ such that more than $\eta t$ of the tests answers are~$0$. Denote by $B''$ the latter event. By (\ref{eq25}), $\mu = \E[X] = e^{-2}t$. Let $\lambda = 1/3-1/(3\Delta)<1$. Then $\eta t>(1+\lambda)\mu$. By (\ref{Chernoff1}) in Lemma~\ref{Chernoff}, we have
$$\Pr[X\ge \eta t]\le \Pr[X\ge (1+\lambda)\mu] \leq e^{-\frac{\lambda^2\mu}{3}}\leq e^{-\frac{(1-\Delta^{-1})^2t}{27e^2}}.$$
Then
\begin{eqnarray*}
\label{AProb2}\Pr[A]&\le& \Pr[B'']\le {n\choose 2\ell/\Delta^2} e^{-\frac{(1-\Delta^{-1})^2t}{27e^2}}\le \left(\frac{e\Delta^2n}{2\ell}\right)^{\frac{2\ell}{\Delta^2}}e^{-\frac{(1-\Delta^{-1})^2t}{27e^2}}\\ &=& \left(\frac{e\Delta^2n}{2\ell}\right)^{\frac{2\ell}{\Delta^2}}e^{-\frac{c'\ell}{27e^2\Delta^2}\ln \frac{c''\Delta^2n}{\ell}}=
\left(\frac{e\Delta^2n}{2\ell}\right)^{\frac{2\ell}{\Delta^2}}\left(\frac{c''\Delta^2n}{\ell}\right)^{-\frac{c'\ell}{27e^2\Delta^2}}
< \frac{1}{2}.
\end{eqnarray*}

\end{proof}

We are now ready to prove Theorem~\ref{thmUBE}.

Let ${\cal A}(\ell,\Delta)$ be the algorithm from Lemma~\ref{fest}. Then, ${\cal A}(\ell,\Delta)$ makes at most
\begin{equation}
\label {Aqueries}
\frac{c\ell}{\Delta^2}\log\frac{\Delta n}{\ell}
\end{equation}
queries for some constant $c$,  and
\begin{enumerate}
\item If ${\cal A}(\ell,\Delta) = 1$, then $d\geq \frac{\ell}{\Delta^2}$.
\item If ${\cal A}(\ell,\Delta) = 0$, then $d \leq \frac{\ell}{\Delta}.$
\end{enumerate}
Consider the algorithm ${\cal T}(n,D,\Delta)$ that runs ${\cal A}(D/\Delta^i,\Delta)$ for all $i=0,\ldots,\lceil \log D/\log \Delta\rceil$. Let $r$ be the minimum integer such that ${\cal A}(D/\Delta^r,\Delta)=1$. Algorithm ${\cal T}(n,D,\Delta)$ then outputs $\hat d=D/\Delta^{r+1}$. See algorithm ${\cal T}$ in Figure~\ref{AlgT}.

\color{black}
\begin{figure}[h!]
  \begin{center}
  \fbox{\fbox{\begin{minipage}{28em}
  \begin{tabbing}
  xxxx\=xxxx\=xxxx\=xxxx\= \kill
    {{${\cal T}$} $(n,D,\Delta)$}\\
  1)  $r\leftarrow 0.$\\
  2)  For each $i=0,1,\ldots, \lceil \log D/\log \Delta\rceil$ do:\\
  \>2.1)  $R\leftarrow {\cal A}(D/\Delta^i,\Delta)$\\
  \> 2.2) If $(R = 1)$ then\\
  \>\> $r \leftarrow i$ \\
  \>\> $\hat d \leftarrow D/\Delta^{r+1}$\\
  \>\> Output ($\hat d$).
    \end{tabbing}
  \end{minipage}}}
  \end{center}
	\caption{Algorithm {$ {\cal T}$} }
	\label{AlgT}
	\end{figure}

We now prove:
\begin{lemma}\label{Reduction} Algorithm ${\cal T}(n,D,\Delta)$ is deterministic non-adaptive that makes
$$O\left(\frac{D}{\Delta^2}\log\left (\frac{\Delta n}{D}\right )\right)$$ tests
and outputs $\hat d$ that satisfies $$\frac{d}{\Delta} \le \hat d\le \Delta d.$$
\end{lemma}
\begin{proof} For $i=0$, if $A(D/\Delta^i,\Delta)=1$ then $d\ge D/\Delta^2$. Then $\hat d=D/\Delta\le \Delta d$ and since $D\ge d$ we also have $\hat d=D/\Delta\ge d/\Delta$.

For $i>0$, if $A(D/\Delta^{i-1},\Delta)=0$ and $A(D/\Delta^i,\Delta)=1$ then $d\le D/\Delta^i$ and $d\ge D/\Delta^{i+2}$. Then
$\hat d=D/\Delta^{i+1}\le \Delta d$ and $\hat d\ge d/\Delta$.

Let $q= \lceil \log D/\log \Delta\rceil$. Let $t$ denote the number of queries performed by algorithm ${\cal T}(n, D,\Delta)$. By~(\ref{Aqueries}), the number of tests is at most
\begin{eqnarray*}
\sum_{i=0}^{q} \frac{cD}{\Delta^i\Delta^2} \log \frac{n\Delta^{i+1}}{D}&\leq&
\frac{cD}{\Delta^2}\sum_{i=0}^{\infty} \frac{1}{\Delta^i} \log \frac{n\Delta^{i+1}}{D}\\
&=& \frac{cD}{\Delta^2}\left( \left(\log\frac{n}{D}\right) \sum_{i=0}^{\infty} \frac{1}{\Delta^i}+(\log\Delta) \sum_{i=0}^{\infty} \frac{i+1}{\Delta^i}\right)\\
&\leq& \frac{cD}{\Delta^2}\left( \frac{\Delta}{\Delta-1}\log\frac{n}{D}+ \frac{\Delta^2}{(\Delta-1)^2}\log\Delta\right).\\
\end{eqnarray*}
For the case when $\Delta=1+\Theta(1)$ we get
$$t = O\left(D\log \frac{n}{D}\right)$$ and for the case when $\Delta=\omega(1)$ we get
$$t = O\left(\frac{D}{\Delta^2} \left(\log\frac{n}{D}+\log \Delta\right)\right).$$
\end{proof}

\section{Lower Bound for Adaptive Deterministic Algorithm}

In this section, we prove the following lower bound.
\begin{theorem}\label{lbound} Any deterministic adaptive group testing algorithm that
given $D>d$, outputs $\hat d$ that satisfies
$d/\Delta\le \hat d \le \Delta d$ must make at least
$$\Omega\left(\frac{D}{\Delta^2}\log \frac{n}{D}\right)$$ queries.
\end{theorem}

For the proof, we use the following from~\cite{BBHKS18}.

\begin{lemma}\label{lb} Let $A$ be a deterministic adaptive algorithm
that for a defective sets $I\subset [n]$ makes the tests
$T^I_1,T^I_2\ldots,T^I_{w(I)}$ and let $s(I)$ be the sequence of answers
to these tests. If $M=|\{s(I)| I\subseteq [n]\}|$ then the test complexity
of $A$ is $\max_Iw(I)\ge \log M$.
\end{lemma}

\ignore{We will use the following result that, in particular, shows that the
upper bound in Lemma~\ref{fest} is tight.}
The following Lemma assists us to prove the result declared by Theorem~\ref{lbound}.
\begin{lemma}\label{lllo} Any deterministic adaptive algorithm such
that, if the number of defectives $d$ is less than or equal $d_1$ it
outputs $0$ and if it is greater than  $d_2$ it outputs $1$, must make
$$\Omega\left(d_1\log\frac{n}{d_2}\right)$$ tests.

In particular, when $d_1=\ell/\Delta^2$ and $d_2=\ell/\Delta$ we get
$$\Omega\left(\frac{\ell}{\Delta^2}\left(\log\frac{n}{\ell}+\log
\Delta\right)\right)$$ tests.
\end{lemma}
\begin{proof} Let $A$ be such algorithm. Let $s(I)$ be the sequence of
answers to the tests of $A$ when the set of defective items is $I$.
Consider a set $I$ of size $d_1$ and let ${\cal J}=\{J\subseteq [n]: |J|=d_1,
s(J)=s(I)\}$. Let $I'=\cup_{J\in {\cal J}}J$. We claim that
$s(I')=s(I)$. Suppose for the contrary, $s(I')\not=s(I)$. Then, since
$I\subseteq I'$, there is a test $Q\subseteq [n]$ that is asked by $A$
that gives answer $0$ to $I$ and $1$ to $I'$. Since $I'\cap
Q\not=\emptyset$, there is a subset $J'\in {\cal J}$ such that $J'\cap
Q\not=\emptyset$ and therefore $Q$ gives answer $1$ to $J'$. Then
$s(J')\not=s(I)$ and we get a contradiction.

Since $s(I')=s(I)$ and algorithm $A$ outputs $0$ to $I$, it also outputs
$0$ to $I'$. Therefore, $|I'|\le d_2$. Therefore
$|{\cal J}|\le N:={d_2 \choose d_1}$. That is, for every possible
sequence of answers $s'$ of the algorithm $A$, there is at most $N$ sets
of size $d_1$ that get the same sequence of answers. Since there are
$L:={n\choose d_1}$ such sets, the number of different sequences of
answers that $A$ might have must be at least $L/N$. By Lemma~\ref{lb}, the
number of tests that the algorithm makes is at least
$$\log\frac{{n\choose d_1}}{{d_2 \choose d_1}}\ge \log
\left(\frac{n}{ed_2}\right)^{d_1}=\Omega\left(d_1\log\frac{n}{d_2}\right).$$
\end{proof}

The conclusions established by Lemma~\ref{lllo} show that the upper bound from Lemma ~\ref{fest} is tight.  Moreover, using these results, we provide  the following proof for Theorem~\ref{lbound}.
\begin{proof}
Let $d_1=D/\Delta^2-1$ and $d_2=D$. For sets of size less than or equal
$d_1$ the algorithm returns $d_1/\Delta\le \hat d\le \Delta d_1$ and for
sets of equal to $d_2$ the algorithm returns $d_2/\Delta< \hat
d\le \Delta d_2$. Since $\Delta d_1<d_2/\Delta$, the above intervals are
disjoint. So, the algorithm can distinguish between sets of size less
that or equal to $d_1$ and sets of size greater than $d_2$. By
Lemma~\ref{lllo} the algorithm must make at least
$$\Omega\left(\frac{D}{\Delta^2}\log\frac{n}{D}\right)$$ tests.
\end{proof}

\section{Polynomial Time Adaptive Algorithm}
In this section, we prove:
\begin{theorem}
\label{thmUBEA}
Let $D$ be some upper bound on the number of defective items $d$ and $\Delta >1$. Then, there is a linear time deterministic adaptive algorithm that makes
$$O\left (\frac{D}{\Delta^2}\left(\log \frac{n}{D} +\log\Delta\right)\right )$$
tests and outputs $\hat{d}$ such that $\frac{d}{\Delta} \leq \hat{d} \leq d\Delta.$
\end{theorem}

We first describe the algorithm. The algorithm gets as an input the set of items $X=[n]$ and splits it into two equally-sized disjoint sets $Q_1$ and $Q_2$. The algorithm asks the queries defined by $Q_1$ and $Q_2$ and proceeds in the splitting process on the sets that yielded positive answers only. We call these sets \emph{defective sets}. As long as the algorithm gets less than $D/\Delta^2$ distinct defective sets, it continues to split and test. Two cases can happen. Either it gets $D/\Delta^2$ defective sets and then the algorithm outputs $\hat d=D/\Delta$, or the number of the defective sets is always less than $D/\Delta^2$ and then, the algorithm finds all the defective items and returns their exact number. The algorithm is given in Figure~\ref{PTA}. The algorithm invokes the procedure ${\bf Split}(X)$ that on an input $X=\{a_1, a_2, \ldots, a_n \}$, it returns the set $W$ where $W:=\{X_1,X_2\}$ such that $X_i\subseteq X$, $X_1 = \{a_1, a_2, \ldots, a_{\left \lfloor{n/2}\right \rfloor }\}$, $X_2 = \{a_{\left \lfloor{n/2}\right \rfloor +1},  \ldots, a_n\}$ if $|X|\geq 2$, and $W:=\{X\}$ otherwise.
\begin{figure}[h!]
  \begin{center}
  \fbox{\fbox{\begin{minipage}{28em}
  \begin{tabbing}
  xxxx\=xxxx\=xxxx\=xxxx\= \kill
  {\texttt{{\bf Adaptive-dEstimate}} $({\cal O}_I, X,\Delta, D)$}\\
  1) $Q\leftarrow {X}, S\leftarrow \emptyset$\\
  2) While $(|Q|\leq D/\Delta^2)$ do:\\
  \> 2.1) For each $Q_i\in Q$ \\
  \>\> $\left \{Q_i^{(1)}, Q_i^{(2)} \right \} \leftarrow{\bf Split}(Q_i) $\\
  \>\> If $(Q_i^{(1)}(I)=1 )$ then $S\leftarrow S\cup \{Q_i^{(1)}\} $\\
  \>\> If $(Q_i^{(2)}(I)=1 )$ then $S\leftarrow S\cup \{Q_i^{(2)}\} $\\
  \> 2.2) If $\forall S_i\in S, |S_i|=1$ \\
  \>\>$ \hat{d}\leftarrow |S|$\\
  \>\> Output ($\hat{d}$)\\
  \> Else\\
  \>\> $Q\leftarrow S, S\leftarrow \emptyset.$\\
  3) $\hat{d}\leftarrow |Q|\cdot \Delta$.\\
  4) Output ($\hat{d}$)\\
    \end{tabbing}
  \end{minipage}}}
  \end{center}
	\caption{Algorithm {\bf Adaptive-dEstimate} to estimate the number of defective items. }
	\label{PTA}
	\end{figure}
	
\begin{lemma}
Algorithm {\bf Adaptive-dEstimate} is a deterministic adaptive algorithm that makes
$$2\frac{D}{\Delta^2}\log{\frac{n\Delta^2}{D}}=O\left (\frac{D}{\Delta^2}\left ( \log{\frac{n}{D}}+\log {\Delta}\right)\right )$$
tests and outputs an estimation $\hat{d}$ such that:
$$\frac{d}{\Delta} \leq \hat{d}\leq d\Delta.$$
\end{lemma}
\begin{proof}
If $d\leq \frac{D}{\Delta^2}$, then the splitting process in step 2 of the algorithm proceeds until  each defective item belongs to a distinct set. Eventually, the condition in step 2.2 is met and the algorithm outputs the exact value of $d$.
If $d> {D/\Delta^2}$, then the splitting process stops when the number of defective sets $|Q|$ exceeds $D/\Delta^2$. The algorithm halts and outputs $\hat{d}=|Q|\Delta$. Obviously, $|Q| \leq d$. Therefore, $\hat{d}=|Q|\Delta \leq d\Delta.$ Moreover, $|Q|> D/\Delta^2\geq d/\Delta^2$ which implies that $\hat{d} \geq d/\Delta.$

The number of iterations cannot exceed $\log n$ iterations. In the first $\log (D/\Delta^2)$ iterations, in the worst case scenario, the algorithm splits its current set $Q_i$ on each iteration into two sets $Q_i^{(1)}$ and $Q_i^{(2)}$ such that $Q_i^{(1)}(I)=Q_i^{(2)}(I) = 1$. Therefore, the number of tests that the algorithm asks over all the first $\log({D}/{\Delta^2})$ iterations is at most
$$\sum_{i=1}^{\log({D}/{\Delta^2})}2^{i} \leq 2\frac{D}{\Delta^2}.$$ Since $|Q| \leq{D}/{\Delta^2}$, in the other $\log n - \log ({D}/{\Delta^2})$ iterations, the algorithm makes at most $2D/{\Delta^2}$ tests each iteration. So, the total number of tests is at most
$$2\frac{D}{\Delta^2}\left (\log n -\log  \frac{D}{\Delta^2}\right ) +2\frac{D}{\Delta^2}=O\left (\frac{D}{\Delta^2}\log \frac{n\Delta^2}{D}\right )$$.
\end{proof}

\section{Polynomial Time Non-Adaptive Algorithm}
In this section, we  show how to use expanders, condensers and extractors to construct deterministic non-adaptive algorithms for defectives number estimation. We prove:

\begin{theorem}
\label{thmUBENA}
Let $D$ be some upper bound on the number of defective items $d$ and $\Delta >1$. Then, there is a polynomial time deterministic non-adaptive algorithm that makes
$$\min\left(D^{o(1)}, 2^{\log^3(\log n)}\right)\cdot \frac{D}{\Delta^2}\log n$$
tests and outputs $\hat{d}$ such that $\frac{d}{\Delta} \leq \hat{d} \leq d\Delta.$
\end{theorem}

\subsection{Algorithms Using Expanders}
Let $G$ be a bipartite graph $G=G(L,R,E)$ with left vertices $L=[n]$, right vertices $R=[m]$ and edges $E\subseteq L\times R$. For each edge $(i,j)\in E$, it holds that the endpoint $i\in L$ and $j\in R$. For a vertex $v\in L$, define $\Gamma(v)$ to be the set of the neighbours of $v$ in $G$ i.e. $\Gamma(v):=\{u\in R| (v,u)\in E\}$. For a subset $S\subseteq L$, we define $\Gamma(S)$ to be  the set of neighbours of $S$, meaning $\Gamma(S):=\cup_{v\in S}\Gamma(v)$. For a vertex $v\in L$, the \emph{degree} of $v$ is defined as $deg(v):=|\Gamma(v)|$. We say that a bipartite graph $G=G(L,R,E)$ is a {\it $(k,a)$-expander $\delta$-regular bipartite graph} if, the degree of every vertex in $L$ is $\delta$, and for every left-subset $S\subseteq L$ of size at most $k$, we have $|\Gamma(S)|\ge a|S|$.

\ignore{ is a {\it $(k,a)$-expander $\delta$-regular bipartite graph} if the degree of each vertex in $L$ is $\delta$ and for every $S\subseteq L$ such that $|S| \leq k$
We will use the sets $L=[n]$ and $R=[m]$ and write $i-j\in E$ to say that an edge with the endpoints $i\in L$ and $j\in R$ is in $E$.
For a vertex $v\in L$ we denote by $\Gamma(v)$ the set of neighbours of $v$
and for set $S\subseteq L$ we denote by $\Gamma(S)$ the set of neighbours of $S$, $\Gamma(S)=\cup_{v\in S}\Gamma(v)$.

We prove the following
}

\begin{lemma} \label{expnA}Let $X=[n]$ be a set of items and $I\subseteq [n]$ is the set of defective items such that $|I|=d$ is unknown to the algorithm. Let $G=G(L,R,E)$ be a $(k,a)$-expander $\delta$-regular bipartite graph with $|L|=n$ and $|R|=m$. Then, there is a deterministic non-adaptive algorithm $A$, such that for $n$ items, it makes $m$ tests and satisfies:
\begin{enumerate}
\item If $|I|< ak/\delta$, then $A$ outputs $0$.
\item If $|I| \geq k$, then $A$ outputs $1$.
\end{enumerate}
\end{lemma}
\begin{proof} For every $j\in R$, we define the test $T^{(j)}=\{i | (i,j)\in E\}$. The number of tests is $|R|=m$. If $|I|\ge k$,  then $|\Gamma(I)|\ge ak$. Therefore, at least $ak$ tests will give positive answer $1$. If $|I|<ak/\delta$, then, since the degree of every vertex in $L$ is $\delta$, we have $|\Gamma(I)|\le \delta|I|<ak$. This shows that, for this case, at most $ak-1$ tests give the answer $1$. Hence, we can distinguish between the two cases.
\end{proof}

Following the same proof of Lemma~\ref{Reduction} with algorithm ${\cal T}$ in Figure~\ref{AlgT}, we have:
\begin{lemma}\label{fred} 
Let $A(\ell, \Delta)$ be a deterministic non-adaptive algorithm such that, for $n$ items, it makes $m(\ell,\Delta)$ tests and satisfies: 
\begin{enumerate}
\item If $|I|< \ell/\Delta^2$, then $A$ outputs $0$.
\item If $|I| \geq \ell/\Delta$, then $A$ outputs $1$.
\end{enumerate}
Then, there is a deterministic non-adaptive algorithm ${\cal T}$ such that, given $D>d$, for $n$ items it makes
$$\sum_{i=0}^{\lceil \log D/\log \Delta\rceil} m\left(\frac{D}{\Delta^i},\Delta\right)$$ tests and outputs $\hat d$ that satisfies $d/\Delta\le \hat d\le \Delta d$.
\end{lemma}

The parameters of the explicit construction of a $(k,a)$-expander $\delta$-regular bipartite graph from ~\cite{CRVW02} are summarised in the following lemma.
\ignore{
M. Capalbo, O. Reingold, S. Vadhan, and A. Wigderson, “Randomness
conductors and constant-degree expansion beyond the degree/2 barrier,”
in Proceedings of the 34th Annual ACM Symposium on Theory of
Computing (STOC), 2002, pp. 659--668.}

\begin{lemma}\label{CRVW02a} For any $k>0$ and $0 <\epsilon <1$, there is an explicit construction of a $(k,a)$-expander $\delta$-regular bipartite graph with
$$m=O(k\delta/\epsilon), \ \ \delta=2^{O(\log^3(\log n/\epsilon))}, \ \ a=(1-\epsilon)\delta.$$
\label{CRVW1}
\end{lemma}
\ignore{By Lemma (\ref{expnA}) and (\ref{CRVW1}) we have:}
We now prove:
\begin{lemma} There is a polynomial time deterministic non-adaptive algorithm that makes
$$\frac{D}{\Delta^2}\cdot 2^{O(\log^3(\log n))} =\frac{D}{\Delta^2}\cdot \mbox{{\rm quasipoly}}(\log n)$$
tests and outputs $\hat d$ that satisfies $$\frac{d}{\Delta} \le \hat d\le \Delta d.$$
\end{lemma}
\begin{proof}
We use the expander in Lemma~\ref{CRVW02a}. Recall that $\Delta=1+\Omega(1)$. Let $r=\min(\Delta,2)$, $\epsilon=1-1/r$ and $k=r\ell/\Delta^2$. Then $a=\delta/r=2^{O(\log^3\log n)}$ and $m=m(\ell,\Delta)=(\ell/\Delta^2)2^{O(\log^3\log n)}$. By Lemma~\ref{expnA}, there is a deterministic non-adaptive algorithm $A$ such that for $n$ items, it makes $m(\ell,\Delta)$ tests and
\begin{enumerate}
\item If $|I|<ak/\delta=\ell/\Delta^2$ then $A$ outputs $0$.
\item\label{c2x} If $|I|\ge k= r\ell/\Delta^2$ then $A$ outputs $1$. 
\end{enumerate}
Algorithm $A$ trivially satisfies the first condition required by Lemma~\ref{fred}.  Consider item \ref{c2x}. If $\Delta<2$ then $r=\Delta$ and then if $|I|\ge k= \ell/\Delta$ then $A$ outputs $1$. If $\Delta>2$ then $r=2$ and then if $|I|\ge k= 2\ell/\Delta^2$ then $A$ outputs $1$. Since $2\ell/\Delta^2<\ell/\Delta$, if $|I|\ge  \ell/\Delta$ then $A$ outputs~$1$.

Now by Lemma~\ref{fred}, there is a deterministic non-adaptive algorithm ${\cal T}$ such that, given $D>d$, for $n$ items, it makes
$$\sum_{i=0}^{\lceil \log D/\log \Delta\rceil} m\left(\frac{D}{\Delta^i},\Delta\right)=\frac{D}{\Delta^2}\cdot 2^{O(\log^3(\log n))} $$ tests and outputs $\hat d$ that satisfies $d/\Delta\le \hat d\le \Delta d$.

\ignore{Let $k=D/\Delta^2$,  $\epsilon = 1/2$ and let $G$ be the expander from Lemma (\ref{CRVW1}). Let ${\cal A}$ be the algorithm from Lemma (\ref{expnA}). We use ${\cal A}$to design an estimation algorithm ${\cal B}$. Algorithm ${\cal B}$ runs ${\cal A}$ with the expander $G$. If ${\cal A}$ output $0$,  then $|I| \geq D/\Delta^2$ and hence, ${\cal B}$ outputs $\hat d = k\cdot \Delta\leq d\Delta$ and it hold that $d/\Delta \leq\hat d \leq d\Delta$.\color{blue} If ${\cal A}$ outputs $1$, then $|I|<D/(2\Delta^2)$. Therefore, in this case, ${\cal B}$ stops and outputs $\hat d = D/(2\Delta^3)$.

Nader: how can I assure that $\hat d < d\Delta$?
\color{black}}
\end{proof}
\subsection{Algorithms Using Extractors and Condensers}\label{lsection}
Extractors are functions that convert weak random sources into almost-perfect random sources. We use these objects to construct a non-adaptive algorithm for estimating $d$. We start with some definitions.

\begin{definition}
Let $X$ be a random variable over a finite set $S$. We say that $X$ has min-entropy at least~$k$ if $Pr[X = x] \le 2^{-k}$ for all $x\in S$.
\end{definition}
\begin{definition}
Let $X$ and $Y$ be random variables over a finite set $S$.
We say that $X$ and $Y$ are $\epsilon-close$ if $\max_{P\subseteq S}|\Pr[X\in P] -\Pr[Y\in P]|\leq \epsilon.$
\end{definition}
We denote by $U_\ell$ the uniform distribution on $\{0,1\}^\ell$. The notations $\Pr_{x\in B}$ or $\E_{x\in B}$ stand for the fact that the probability and the expectation are taken when  $x$ is chosen randomly uniformly from $B$.
\begin{definition}
A function $F:\{0,1\}^{\hat n}\times \{0,1\}^{\hat t}\to \{0,1\}^{\hat m}$ is a  $k\to_\epsilon k'$ condenser if for every $X$ with min-entropy at least $k$ and $Y$ uniformly distributed on $\{0,1\}^{\hat t}$, the distribution of $(Y,F(X,Y))$ is $\epsilon$-close to a distribution $(U_{\hat t},Z)$ with min-entropy $\hat t+k'$. A condenser is called $(k,\epsilon)$-lossless condenser if $k'=k$. A condenser is called $(k,\epsilon)$-extractor if $\hat m=k'$.
\end{definition}
Let $\hat N=\{0,1\}^{\hat n}, \hat T=\{0,1\}^{\hat t}$ and $\hat M=\{0,1\}^{\hat m}$, and let $F: \hat N\times \hat T \to \hat M$ be a $k\to_\epsilon k'$ condenser. Consider the $2^{\hat t}\times 2^{\hat n}$ matrix ${\cal M}$ induced by $F$.  That is, for  $r\in \hat T$ and $s\in \hat N$, the entry ${\cal M}_{r,s}$ is equal to $F(s,r)$. For $s\in \hat N$, let ${\cal M}^{(s)}$ be the $s$th column of ${\cal M}$. Then, ${\cal M}^{(s)}_r={\cal M}_{r,s}=F(s,r)$.
\begin{definition}
Let $\Sigma$ be a finite set. An $n$-mixture over $\Sigma$ is an $n-$tuple ${\cal S}:= (S_1, \cdots, S_n)$ such that for all $i\in[n]$, $S_i\subseteq \Sigma$.
\end{definition}
Using these definitions and notations, we restate the result proved by Cheraghchi~\cite{MC08} (Theorem 9) in the following lemma.

\begin{lemma}\label{Mahdi} Let $F:\{0,1\}^{\hat n}\times \{0,1\}^{\hat t}\to \{0,1\}^{\hat m}$ be a $k\to_\epsilon k'$ condenser. Let ${\cal M}$ be the matrix induced by $F$. Then, for any $2^{\hat t}-$mixture ${\cal S}=(S_1, \cdots, S_{2^{\hat t}})$ over $\hat M:=\{0,1\}^{\hat m}$, the number of columns $s$ in ${\cal M}$ that satisfies
$$\underset{{r\in\hat T}}{\Pr}[{\cal M}^{(s)}_r\in S_r]>\frac{\E_{r\in\hat T}[|S_r|]}{2^{k'}}+\epsilon$$ is less than $2^k$.
\end{lemma}

Equipped with Lemma~\ref{Mahdi}, we prove: 
\begin{lemma}\label{Main2} If there is a $k\to_\epsilon k'$ condenser $F:\{0,1\}^{\hat n}\times \{0,1\}^{\hat t}\to \{0,1\}^{\hat m}$ then, there is a deterministic non-adaptive algorithm ${\cal A}$ for $n=2^{\hat n}$ items that makes $m=2^{\hat t+\hat m}$ tests and satisfies the following.
\begin{enumerate}
\item If the number of defectives is less than $(1-\epsilon)2^{k'}$ then ${\cal A}$ outputs $0$.
\item If the number of defectives is greater than or equal $2^k+1$ then ${\cal A}$ outputs $1$.
\end{enumerate}
\label{cnstT}
\end{lemma}
\begin{proof}
Consider the matrix ${\cal M}$ induced by the condenser $F$ as explained above. We define the test matrix ${\cal T}$ from ${\cal M}$ as follows. Let $x\in \{0,1\}^{\hat m}$. Define $e(x)\in \{0,1\}^{2^{\hat m}}$ such that $e(x)_y=1$ if and only if $x=y$, where the bits in $e(x)$ are indexed by the elements of $\{0,1\}^{2^{\hat m}}$. Each row  $r$ in the matrix ${\cal M}$ is replaced by $2^{\hat m}$ rows (in ${\cal T}$) such that in each entry ${\cal M}_{r,s}\in \{0,1\}^{\hat m}$ is replaced by the column vector $e({\cal M}_{r,s})^T\in \{0,1\}^{2^{\hat m}}$. The rows of the matrix ${\cal T}$ are indexed by $\hat T \times \hat M$. Let ${\cal T}^{(i)}$ denote the $i$th column of ${\cal T}$. Therefore, for $r\in \hat T$ and $j\in \hat M$, the row $(r,j)$ in the matrix ${\cal T}$ is denoted by ${\cal T}_{(r,j)}$. \ignore{  as described in Figure (\ref{digCnst}). }
Moreover, the $i$th entry of the row ${\cal T}_{(r,j)}$ is denoted by ${\cal T}_{(r,j),i}$ and ${\cal T}_{(r,j),i}={\cal T}^{(i)}_{(r,j)}=1$ if and only if ${\cal M}_{r,i}=j$.   The size of the test matrix ${\cal T}$ is $m\times n$. 

Let the defective elements be $s_{i_1},\ldots,s_{i_\ell}$ and let $y\in \{0,1\}^{m}$  indicate the tests result. Then, $y$ is equal to ${\cal T}^{(s_{i_1})}\vee \cdots\vee {\cal T}^{(s_{i_\ell})}$. Let ${\cal S}=(S_r)_{r\in \hat T}$ be a $2^{\hat t}-$mixture over $\{0,1\}^{{\hat m}}$ where for all $r\in \hat T$, $S_r=\{j\in \{0,1\}^{{\hat m}} |y_{(r,j)}={\cal T}^{(s_{i_1})}_{(r,j)}\vee \cdots\vee {\cal T}^{(s_{i_\ell})}_{(r,j)} =1\}$. It is easy to see that:

\begin{enumerate}
\item $|S_r|\le \ell$. This is because, by the definition of $S_r$,  $j\in S_r$ if and only if $y_{(r,j)} =1$. The entry $y_{(r,j)}$ gets the value $1$ if at least one of the entries ${\cal T}_{(r,j)}^{(s_{i_1})}, \cdots,{\cal T}_{(r,j)}^{(s_{i_\ell})}$ is $1$. Any row in ${\cal T}^{(s_{i_1})}, \cdots, {\cal T}^{(s_{i_\ell})}$ has exactly one entry that is equal to $1$ in all the $2^{\hat m}$ rows indexed by $r$. Hence, each row can cause one item to be inserted to $S_r$.
\item For any $s_{i_j}\in \{s_{i_1},\ldots,s_{i_\ell}\}$,  we have $\Pr_{r\in \hat T}[{\cal M}^{(s_{i_j})}_r\in S_r]=1$
\item Given the matrix ${\cal M}$, its test matrix ${\cal T}$ and the observed result $y$, for any column $s$ the probability $\Pr_{r\in \hat T}[{\cal M}^{(s)}_r\in S_r]$  can be easily computed.
\end{enumerate}

If the number of defectives is less than $(1-\epsilon)2^{k'}$ then, by Lemma~\ref{Mahdi}, all columns, except for at most $2^k$ columns, satisfy
$$\underset{{r\in\hat T}}{\Pr}[{\cal M}^{(s)}_r\in S_r]\le \frac{\E_{r\in\hat T}[|S_r|]}{2^{k'}}+\epsilon< \frac{\E_{r\in\hat T}[(1-\epsilon)2^{k'}]}{2^{k'}}+\epsilon=1.$$
So for less than $2^k+1$ columns we have $\Pr_{r\in \hat T}[{\cal M}^{(s)}_r\in S_r]=1$.
If the number of defectives is greater than or equal $2^{k}+1$, then for the columns of the defectives we have $\Pr_{r\in \hat T}[{\cal M}^{(s)}_r\in S_r]=1$. So for more than $2^k$ columns we have $\Pr_{r\in \hat T}[{\cal M}^{(s)}_r\in S_r]=1$.

\ignore{

 such that corresponds to ${\cal S}=S_{1}\times S_{2}\times \cdots \times S_{2^{\hat t}}\subset \Sigma^{2^{\hat t}}$ where $S_r=\{{\cal M}_r^{(s_1)},\ldots, {\cal M}_r^{(s_\ell)}\}$. It is easy then to see than

the  For each entry $x\in \{0,1\}^{\hat m}$ to the zero-one column $e_x$ that is equal to $0$ except for the entry $x$ (equal $1$). The size of the matrix is $m\times n$. Notice that if the defectives are $s_1,\ldots,s_\ell$ then the column of results of the tests is ${\cal M}^{(s_1)}\vee \cdots\vee {\cal M}^{(s_\ell)}$ corresponds to ${\cal S}=S_1\times S_2\times \cdots \times S_{2^{\hat t}}\subset \Sigma^{2^{\hat t}}$ where $S_r=\{{\cal M}_r^{(s_1)},\ldots, {\cal M}_r^{(s_\ell)}\}$. It is easy then to see than}
\end{proof}

The following Lemma summarises the state of the art result due to Guruswami et. al. ~\cite{GUV07} on explicit construction of expanders.

\begin{lemma}\label{GUV07} For all positive integers $\hat n,k$ such that $\hat n\geq k$, and all $\epsilon >0$, there is an explicit $(k,\epsilon)$ extractor $F:\{0,1\}^{\hat n}\times \{0,1\}^{\hat t}\to \{0,1\}^{\hat m}$ with $\hat t = \log {\hat n}+O(\log{k}\log{(k/\epsilon)})$ and $\hat m =k'= k -2\log {1/\epsilon}-c$ for some constant $c$.
\end{lemma}

We now prove:

\begin{lemma}
\label{thrmDNonAdp}
There is a constant $C$ such that for every $\Delta>C$, there is a polynomial deterministic non-adaptive algorithm that estimates the number of defective items in a set of $n$ items up to a multiplicative factor of $\Delta$ and asks $$O\left(\frac{D^{1+o(1)}}{\Delta^2}\log {n}\right )$$ queries.
\end{lemma}
\begin{proof}
We use the notations from Lemma~\ref{GUV07}. Let $C=27\cdot 2^{c-2}$. We choose $\epsilon=2/3$ and $k'$ such that $(1-\epsilon)2^{k'}=\ell/\Delta^2$. Then $$2^k=2^{k'+2\log(1/\epsilon)+c}=27\cdot 2^{c-2}\frac{\ell}{\Delta^2}< \frac{\ell}{\Delta}$$

By Lemma~\ref{Main2}, there is a deterministic non-adaptive algorithm ${\cal A}$ for $n=2^{\hat n}$ items that makes 
$$m=2^{\hat t+\hat m}=\hat n 2^{O(\log k \log (k/\epsilon))}\frac{\ell}{(1-\epsilon)\Delta^2}=2^{\log^2\log(\ell/\Delta)}\frac{\ell}{\Delta^2}\log n$$ tests that satisfies the following:
\begin{enumerate}
\item If the number of defectives is less than $(1-\epsilon)2^{k'}={\ell}/{\Delta^2}$ then ${\cal A}$ outputs $0$.
\item If the number of defectives is greater than or equal $2^k+1$ then ${\cal A}$ outputs $1$ and, since $2^k<\ell/\Delta$, in particular, if the number of defectives is greater than or equal $\ell/\Delta$ then ${\cal A}$ outputs $1$.
\end{enumerate}
By Lemma~\ref{fred}, the result follows.
\end{proof}
\ignore{
they construct a $(k,\epsilon)$-extractor with
$$\hat m=k-2\log(1/\epsilon)-O(1),\ \hat t=\log \hat n+O(\log k\log (k/\epsilon))$$

For an extractor $F:\{0,1\}^{\hat n}\times \{0,1\}^{\hat t}\to \{0,1\}^{\hat m}$, the construction described in Lemma (\ref{cnstT}) implies that the dimension of the test matrix is ${\cal T}$ is $2^{\hat t +\hat m}\times 2^{\hat n}$. Therefore, for group testing over $n$ items we need to choose an extractor with $\hat n = \log {n}$. Let $F:\{0,1\}^{\hat n}\times \{0,1\}^{\hat t}\to \{0,1\}^{\hat m}$ be the extractor from Lemma (\ref{GUV07}) where $\hat n= \log{n}$. In the paper ~\cite{GUV07} they give an explicit construction for lossless condenser (Theorems 4.3.and 4.4).
}

A similar work by Capalbo et.al. ~\cite{CRVW02} gives an explicit constrction of a lossless condenser is summarised in the following lemma:
\begin{lemma}
\label{CRVW02} For all positive integers $\hat n, k$  and all $\epsilon >0$, there is an explicit lossless condenser $F:\{0,1\}^{\hat n}\times \{0,1\}^{\hat t}\to \{0,1\}^{\hat m}$ with $\hat t = O(\log^3(\hat n/\epsilon))$ and $\hat m = k+\log(1/\epsilon)+O(1)$.
\end{lemma}
The construction from Lemma~\ref{CRVW02} yields a result that is similar to the one established in Lemma~\ref{thrmDNonAdp}.
\ignore{
In M. Capalbo, O. Reingold, S. Vadhan, and A. Wigderson, “Randomness
conductors and constant-degree expansion beyond the degree/2 barrier,”
in Proceedings of the 34th Annual ACM Symposium on Theory of
Computing (STOC), 2002, pp. 659--668.

They gave a lossless condenser with
$$\hat t=O(\log^3(\hat n/\epsilon))\mbox{ and}\  \hat m=k+\log(1/\epsilon)+O(1).$$

This gives the same result as the above result we get from expander.

I suggest you look at the paper:

In Mahdi Cheraghchi:
Noise-resilient group testing: Limitations and constructions. Discret. Appl. Math. 161(1-2): 81-95 (2013)

and search for other extractors and condensers and see what they give}

\vskip 0.2in
\bibliography{Ref}
\bibliographystyle{plain}
\end{document}